\documentclass[runningheads]{llncs}
\usepackage{graphicx}
\usepackage{amsmath,amssymb,mathrsfs}
\usepackage{floatflt}
\usepackage{array}
\usepackage[linesnumbered,vlined,ruled]{algorithm2e}
\begin{document}
\title{Oblivious Robots Performing Different Tasks on Grid Without Knowing their Team Members\thanks{The preliminary version of this paper appeared in the conference ICARA 2023}}
\titlerunning{Performing different tasks by robots}
%
\author{Satakshi Ghosh\inst{1}\orcidID{0000-0003-1747-4037}  \and Avisek Sharma\inst{1}\orcidID{0000-0001-8940-392X} \and Pritam Goswami\inst{1}\orcidID{0000-0002-0546-3894} \and
Buddhadeb Sau\inst{1}\orcidID{0000-0001-7008-6135}}
\authorrunning{S. Ghosh et al.}
%
\institute{Jadavpur University, Kolkata, 700032\\
\email{\{satakshighosh.math.rs,aviseks.math.rs,pritamgoswami.math.rs, buddhadeb.sau\}@jadavpuruniversity.in}}
\maketitle              
\begin{abstract}
Two fundamental problems of distributed computing are Gathering and Arbitrary pattern formation (\textsc{Apf}). These two tasks are different in nature as in gathering robots meet at a point but in \textsc{Apf} robots form a fixed pattern in distinct positions. 
 
In most of the current literature on swarm robot algorithms, it is assumed that all robots in the system perform one single task together. Two teams of oblivious robots deployed in the same system and different teams of robots performing two different works simultaneously where no robot knows the team of another robot is a new concept in the literature introduced by Bhagat et al. [ICDCN'2020].

In this work, a swarm of silent and oblivious robots are deployed on an infinite grid under an asynchronous scheduler. The robots do not have access to any global coordinates. Some of the robots are given input of an arbitrary but unique pattern. The set of robots with the given pattern is assigned the task of forming the given pattern on the grid. The remaining robots are assigned with the task of gathering to a vertex of the grid (not fixed from earlier and not any point where a robot that is forming a pattern terminates). Each robot knows to which team it belongs, but can not recognize the team of another robot. Considering weak multiplicity detection, a distributed algorithm is presented in this paper which leads the robots with the input pattern into forming it and other robots into gathering on a vertex of the grid on which no other robot forming the pattern, terminates.


\end{abstract}

\section{Introduction}
In swarm robotics, robots solving some tasks with minimum capabilities is the main focus of interest. In the last two decades, there has been huge research interest in robots working with coordination problems. It is not always easy to use robots with strong capability in real-life applications, as the making of these robots is not at all cost-effective. If a swarm of robots with minimum capabilities can do the same task then it is effective to use swarm robots rather than using robots with many capabilities, as the making of these robots in the swarm is very much cheaper than making robots with many capabilities. Also, it is very easy to design a robot of a swarm due to the fact that they have minimum capabilities. Depending on these capabilities there are generally four types of robot models. These models are $\mathcal{OBLOT}$, $\mathcal{FSTA}$, $\mathcal{FCOM}$ and $\mathcal{LUMI}$. In each of these models, robots are assumed to be autonomous (i.e the robots do not have any central control), identical (i.e the robots are physically indistinguishable), and anonymous (i.e the robots do not have any unique identifiers). Furthermore in the $\mathcal{OBLOT}$ model, the robots are silent (i.e there is no means of communication between the robots) and oblivious (i.e the robots do not have any persistent memory to remember their previous state), in $\mathcal{FSTA}$ model the robots are silent but not oblivious, in $\mathcal{FCOM}$ model the robots are oblivious but not silent and in $\mathcal{LUMI}$ model robots are neither silent nor oblivious. The robots after getting activated operate in a \textsc{Look-Compute-Move} (\textsc{Lcm}) cycle. In the \textsc{Look} phase a robot takes input from its surroundings and then with that input runs the algorithm in \textsc{Compute} phase to get a destination point as an output. The robot then goes to that destination point by moving in the \textsc{Move} phase. The activation of the robots is controlled by a scheduler. There are mainly three types of schedulers considered in the literature. In a synchronous scheduler, time is divided into global rounds. In a fully synchronous (\textsc{FSync}) scheduler, each robot is activated in all the rounds and executes \textsc{Lcm} cycle simultaneously. In a semi-synchronous (\textsc{SSync}) scheduler all robots may not get activated in each round. But the robots that are activated in the same round execute the \textsc{Lcm} cycle simultaneously. Lastly in the asynchronous (\textsc{ASync}) scheduler, there is no common notion of time, a robot can be activated at any time. There is no concept of global rounds. So there is no assumption regarding synchronization.

The Gathering and Arbitrary pattern formation are two vastly studied problems by researchers in the field of swarm robot algorithms. These are one of the fundamental tasks which can be done by autonomous robots in different settings. In the gathering problem, $n$ number of robots initially positioned arbitrarily meets at a point not fixed from earlier within finite time. It is not always easy to meet at a point with very weak robots in the distributed system. Similarly, the Arbitrary pattern formation problem is such that robots have to form a given pattern that is given as input to the robots within a finite time. In literature, there are several works that have considered either gathering or arbitrary pattern formation problem separately. But none of those works consider robots deployed in the same environment working on two different tasks. The environment of robot swarm needs periodic maintenance for making the environment robust from faults and some other factors. So if the same robot swarm deployed in the environment can do the maintenance apart from doing the specific task assigned to them, it would be more cost-effective. From this motivation and practical interest, in \cite{BhagatS20} authors first studied the problem where two teams of oblivious robots work on two different tasks, namely gathering and circle formation on a plane. Here the crucial part is that a robot knows to which team it belongs, but it can not recognize another robot's team. The novelty of the problem would have gone away a bit if the robots are luminous and gathering team robots put a color on a light to indicate which team they belong to. But the main motivation of this problem is to extend the work for more than two different tasks and for assigning different colors for different tasks would make the number of colors unbounded. For this reason, it is convenient to solve the problem with the least possible capabilities for the robot swarm. Also $\mathcal{OBLOT}$ model is more self-stabilized and fault tolerant. For these reasons, in our work, we also consider robots to be oblivious. Now it is challenging to design a distributed algorithm by which two different teams of oblivious robots can do two different tasks simultaneously on a discrete domain because, in any discrete graph, the movements of robots become restricted. So avoiding collision becomes a great challenge. 

In our work, we have provided a collision-less distributed algorithm following which two teams of oblivious robots with weak multiplication detection ability can do two different tasks namely gathering and arbitrary pattern formation simultaneously on an infinite grid under the asynchronous scheduler. Here we assume that the initial configuration is asymmetric. 

\section{Earlier Works}
Arbitrary pattern formation and gathering of robots are two hugely studied problems in the distributed system. In literature, there are many works on these two problems in various settings. Arbitrary pattern formation on a plane was first studied by Suzuki and Yamashita \cite{Suzuki96}. For the grid network, the arbitrary pattern formation was first studied in \cite{BoseAKS20} in the $\mathcal{OBLOT}$ model with full visibility. But in this paper, the algorithm is not move optimal. So in \cite{GHOSH2022} authors have shown on an infinite grid a move optimal \textsc{Apf} algorithm under $\mathcal{OBLOT}$ model and a time optimal \textsc{Apf} algorithm under $\mathcal{LUMI}$ model. Later in \cite{cicerone20} authors studied the \textsc{Apf} on a regular tessellation graph. In an infinite grid, the arbitrary pattern formation problem was studied in \cite{manash22} with opaque robots and in \cite{KunduGGS22} with fat robots. Similarly, for gathering there are many works in an infinite grid. In \cite{d2012gathering}, authors have shown that the gathering is possible on a grid without multiplicity detection. There are many other works (\cite{fischer2017,bhagat2020,cord2016,das2019}) where authors have solved gathering on an infinite rectangular grid with various assumptions. In \cite{pritam22} authors solved the gathering of robots on an infinite triangular grid in \textsc{Ssync} scheduler under one axis agreement and with one-hop visibility. But in all these works they only consider that all robots perform one individual task. But in \cite{BhagatS20} authors first showed that two specific but different tasks can be done by robots simultaneously on a plane. They showed that gathering and circle formation can be done by oblivious robots in a plane simultaneously with two different teams of robots using one-axis agreement. 

So here we are interested to show that two different tasks namely gathering and the arbitrary pattern formation of oblivious robots can be done on an infinite grid simultaneously under an asynchronous scheduler without any axis agreement. This paper first time deals with this problem under a discrete environment. In \cite{cicerone20}, authors considered multiplicity points in the target pattern that needs to be formed. So their work is also capable of forming an arbitrary pattern along with an additional multiplicity point. But by following their algorithm a robot on team gathering might end up forming a pattern and also a robot on team arbitrary pattern formation might end up on the multiplicity point, which is not the required solution to our problem. The algorithm proposed in paper \cite{GHOSH2022} uses similar technique to select $head$ and $tail$ robots and also the forming of the fixed pattern. But in \cite{GHOSH2022} the robots have no multiplicity detection capability and throughout the algorithm, no multiplicity points will form. So from the existing previous algorithms (as in \cite{GHOSH2022,cicerone20}), two different tasks are not trivially solved by robots as a robot does not know in which team another robot belongs.
\section{Model and Problem Statement}
\paragraph*{Robots} Robots are anonymous, identical, and oblivious, i.e. they have no memory of their past rounds. They can not communicate with each other. There are two teams between the robots. One is $\mathcal{T}_{Apf}$ and the other team is $\mathcal{T}_g$. A robot $r$ only knows that in which team it belongs to between this two. But a robot can not identify to which team another robot belongs. All robots are initially in distinct positions on the grid. The robots can see the entire grid and all other robots' positions which means they have global visibility. This implies the robots are transparent and hence the visibility of a robot can not be obstructed by another robot. Robots have no access to any common global coordinate system. They have no common notion of chirality or direction. A robot has its local view and it can calculate the positions of other robots with respect to its local coordinate system with the origin at its own position. There is no agreement on the grid about which line is  $x$ or $y$-axis and also about the positive or negative direction of the axes. As the robots can see the entire grid, they will set the axes of their local coordinate systems along the grid lines. Also, robots have weak multiplicity detection capability, which means a robot can detect a multiplicity point but cannot count the number of robots present at a multiplicity point.

\paragraph*{Look-Compute-Move cycles}
An active robot operates according to the Look-Compute-Move cycle. In each cycle a robot takes a snapshot of the positions of the other robots according to its own local coordinate system (\textsc{Look}); based on this snapshot, it executes a deterministic algorithm to determine whether to stay put or to move to an adjacent grid point (\textsc{Compute}); and based on the algorithm the robot either remain stationary or makes a move to an adjacent grid point (\textsc{Move}). When the robots are oblivious they have no memory of past configurations and previous actions. After completing each \textsc{Look-Compute-Move} cycle, the contents in each robot's local memory are deleted.

\paragraph*{Scheduler} 
We assume that robots are controlled by a fully asynchronous adversarial scheduler (\textsc{ASync}). The robots are activated independently and each robot executes its cycles independently. This implies the amount of time spent in \textsc{Look}, \textsc{Compute}, \textsc{Move}, and inactive states are finite but unbounded, unpredictable, and not the same for different robots. The robots have no common notion of time.

\paragraph*{Movement} In discrete domains, the movements of robots are assumed to be instantaneous. This implies that the robots are always seen on grid points, not on edges. However, in our work, we do not need this assumption. In the proposed algorithm, we assume the movements are to be instantaneous for simplicity. However, this algorithm also works without this. The movement of the robots is restricted from one grid point to one of its four neighboring grid points.
\subsection{Problem Description}
In this work, we define a problem on an infinite grid where $n$ oblivious, identical, autonomous robots are dispersed on the vertices of the infinite grid. In \cite{BhagatS20} they solved two conflicting tasks by robots on a plane where one team of robots gather at a point and another team of robots form a circle on the plane. But in this work robots are on an infinite grid and they solve two different distributed problems. Here are two teams of oblivious robots where one team is $\mathcal{T}_g$ and the other team is $\mathcal{T}_{Apf}$. The goal of the robots of ${\mathcal{T}_g}$ is to meet all the robots of this team to a point on an infinite grid. Another team is $\mathcal{T}_{Apf}$ where the goal of this team is to form a particular pattern that is given as input. The next section provides the algorithm for solving this problem. 
\section{The Main Algorithm}
\subsection{Global Coordinate Agreement}\label{global}
Here we have to first fix the global coordinate system and then we aim to gather the $\mathcal{T}_g$ robots to the origin and form the arbitrary pattern by $\mathcal{T}_{Apf}$ robots with respect to the coordinate $(0,2)$. So let us consider an infinite grid $G$ as a cartesian product $P \times P$, where $P$ is an infinite (from both sides) path graph. The infinite grid $G$ is embedded in the Cartesian Plane $R^2$. In this work, some robots will gather at a point and other robots will form an arbitrary pattern on the grid. Here we are assuming that the initial configuration is asymmetric. A robot can form a local coordinate system aligning the axes along the grid lines but the robots do not have an access to any global coordinate system even. To form the target pattern the robots need to reach an agreement on a global coordinate system. In this subsection, we will provide the details of the procedure that allows the robots to reach an agreement on a global coordinate system.
 \begin{figure}[t!]
\centering
\includegraphics[width=.7\linewidth]{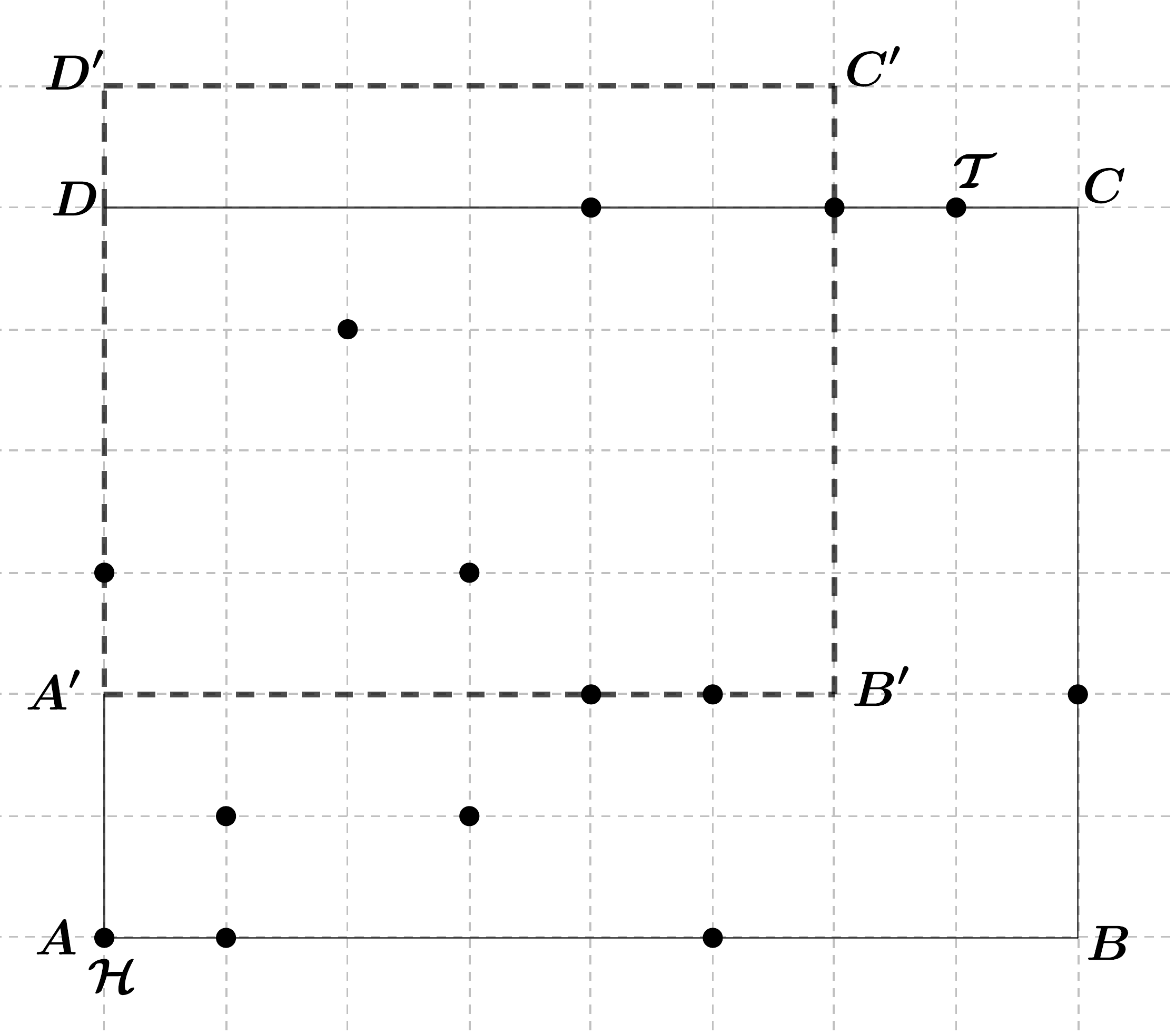}
     \caption{$ABCD$ is the smallest enclosing rectangle of initial configuration. $\mathcal{H}$ and $\mathcal{T}$ are head and tail of the configuration. $A$ is the origin and the target configuration will be embedded with respect to $(0.2)$. $A'B'C'D'$ is the smallest enclosing rectangle of target configuration}
     \label{Fig:n}
\end{figure}     

For a given configuration ($\mathcal{C}$) formed by the robots, let the smallest enclosing rectangle, denoted by $s.rect(\mathcal{C})$, be the smallest grid-aligned rectangle that contains all the robots. Suppose the $s.rect$ of the initial configuration $\mathcal{C}_I$ is a rectangle $\mathcal{R}=ABCD$ of size $m\times n$, such that $m>n>1$. Let $|AB|=n$. Then consider the binary string $\{p_i\}$ associated with a corner $A$, $\lambda_{AB}$ as follows. Scan the grid from $A$ along the side $AB$ to $B$ and sequentially all grid lines of $s.rect(\mathcal{C_I})$ parallel to $AB$ in the same direction. And $p_i=0$, if the position is unoccupied and $p_i=1$ otherwise. Similarly construct the other binary strings $\lambda_{BA}$, $\lambda_{CD}$ and $\lambda_{DC}$. Since the initial configuration is asymmetric we can find a unique lexicographically largest string. If $\lambda_{AB}$ is the lexicographically largest string, then $A$ is called the leading corner of $\mathcal{R}$.

Next, suppose $\mathcal{R}$ is an $m\times m$ square, then consider the eight binary strings $\lambda_{AB}$, $\lambda_{BA}$, $\lambda_{CD}$, $\lambda_{DC}$, $\lambda_{BC}$, $\lambda_{CB}$, $\lambda_{AD}$, $\lambda_{DA}$. Again since the initial configuration is asymmetric, we can find a unique lexicographically largest string among them. Hence we can find a leading corner here as well.

Next, let $\mathcal{C}_I$ be a line $AB$, we will have two strings $\lambda_{AB}$ and $\lambda_{BA}$. Since $\mathcal{C}_I$ is asymmetric then $\lambda_{AB}$ and $\lambda_{BA}$ must be distinct. If $\lambda_{AB}$ is lexicographically larger than $\lambda_{BA}$, then we choose $A$ as the leading corner.

Now for either case, if $\lambda_{AB}$ is the lexicographically largest string then the leading corner $A$ is considered as the origin, and the $x-$ axis is taken along the $AB$ line. If $\mathcal{C}_I$ is not a line then the $y-$ axis is taken along the $AD$ line. If $\mathcal{C}_I$ is a line then the $y-$ coordinate of all the positions of robots is going to be zero and in this case, the $y-$ axis will be determined later.  
For any given asymmetric configuration $\mathcal{C}$ if $\lambda_{AB}$ is the largest associated binary string to $\mathcal{C}$ then the robot causing the first non-zero entry in $\lambda_{AB}$ is called $head$ let $\mathcal{H}$ and the robot causing last non zero entry in $\lambda_{AB}$ is called as $tail$ let $\mathcal{T}$. We denote the $i^{th}$ robot of the $\lambda_{AB}$ string as $r_{i-1}$. A robot other than the $head$ and $tail$ is called \textit{\textbf{inner robot}}. Further we denote  $\mathcal{C}'=\mathcal{C}\setminus\{tail\}$ and $\mathcal{C}''=\mathcal{C}\setminus\{tail, head\}$ and $\mathcal{C}'''=\mathcal{C}\setminus\{head\}$.

Let $\mathcal{C}_{target}$ be the target configuration for the $\mathcal{T}_{Apf}$ robots and $s.rect(\mathcal{C}_{target})=\mathcal{R}_{target}$. Let $\mathcal{R}_{target}$ is a rectangle of size $M\times N$ with $M\ge N$. We can calculate the binary strings associated with corners in the same manner as previously. $\mathcal{C}_{target}$ is expressed in the coordinate system with respect to the point $(0,2)$, where $(0,2)$ will be the leading corner. Let the $A'B'C'D'$ be the smallest rectangle enclosing the target pattern with $A'B'\le B'C'$. Let $\lambda_{A'B'}$ be the largest (may not be unique) among all other strings. Then the target pattern is to be formed such that $A$ is the origin and pattern embedded with respect to the position $(0,2)$, $A'B'$ direction is along the positive $x$ axis and $A'D'$ direction is along the positive $y$ axis. If the target pattern has symmetry then we have to choose any one among the larger string and fixed the coordinate system. So as previously said $head_{target}$ will be the first one and $tail_{target}$ will be the last one in the $s.rect$ of $\mathcal{C}_{target}$. Also, we define
$\mathcal{C}_{final}$ is the configuration when all $\mathcal{T}_g$ robots are at same point and $\mathcal{C}_{target}$ configuration is formed.
$\mathcal{C}_{final}'=\mathcal{C}_{final}\setminus\{tail_{target}\}$, $\mathcal{C}_{final}''=\mathcal{C}_{final}\setminus\{head_{target},tail_{target}\}$, $\mathcal{C}_{final}'''=\mathcal{C}_{final}\setminus\{head_{target}\}$.  
We denote the $head_{target}$ position as $t_1$ and $tail_{target}$ position as $t_{k}$. Let $H_i$ be the horizontal line having the height $i$ from the x-axis. Let for each $i$ there are $p(i)$ target positions on $H_i$. We denote the target positions of $H_0$ as $t_1$......$t_{p(0)-1}$ from left to right. For $H_1$ we denote the target positions as $t_{p(0)}$ to $t_{p(0)+p(1)-1}$ from right to left. For $H_2$ we denote the target positions as $t_{p(0)+p(1)}$ to $t_{p(0)+p(1)+p(2)-1}$ from right to left. Similarly, we can denote all other target positions on $H_i$, $i>0$ except $tail_{target}$.

\subsection{Brief Discussion of Algorithm}
Let the initial configuration is $\mathcal{C_I}$, the final configuration is $\mathcal{C}_{final}$. Robots are operating on an infinite grid. There are two teams of robots where one is $\mathcal{T}_g$ and another one is $\mathcal{T}_{Apf}$ (let us assume that $|\mathcal{T}_g| > 2$). $\mathcal{T}_g$ robots will gather at the same point on the grid and the robots of $\mathcal{T}_{Apf}$ will form a pattern on the grid. Note that the gathering point will not be a target point of the target pattern. 
A robot only knows in which team it belongs between $\mathcal{T}_g$ and $\mathcal{T}_{Apf}$. A robot has no information about to which team other robots belong. Robots first fix the global coordinate system. The target will form with respect to the point (0,2) and the gathering will occur at the origin. When a robot awakes up it first calculates the $head$ robot and $tail$ robot. In the first three stages, the $head$ robot will move to the origin and the $tail$ robot will expand the smallest enclosing rectangle for maintaining the asymmetry of the configuration. Then in the stage~4 all the inner robots now move to the x-axis and make a line on the x-axis. Note that a line is called \textit{\textbf{compact line}} when there is no empty grid point between two robots. So robots on the x-axis first make the line compact then one by one inner robot from upward horizontal lines move down to the x-axis. After this in stage~5 the robot which is not on a line say $r$, will move to its closest endpoint of the line if it $\in \mathcal{T}_g$ or it will move one step upward if it belongs to $\mathcal{T}_{Apf}$. In the next stage after a multiplicity point is formed or calculating the position of the $tail$, robots on the x-axis move to the fixed target positions which are either the origin or the target positions. Robots will move from right to left with respect to the position of the $head$ sequentially. As all inner robots are on the x-axis so the robot $\in \mathcal{T}_g$ can always find the neighboring grid lines empty, so the robot can move to the origin by choosing any path to the origin if all the robots are on a line. If a robot can see the $tail$ robot then it will choose the neighboring line of the x-axis in the direction of the $tail$ for its movement. In this way, one by one all inner robots move to their fixed target positions. Next if $tail$ sees that all inner robots move to their target positions it will move to their fixed position. In the last stage if the $head$ robot is in the gathering team then it will not move but if it is in the \textsc{Apf} team then when without its position all the pattern formation is done it will move to its position. Note that within the algorithm no two robots collide without the gathering point and the coordinate system remains unchanged. Within finite time all $\mathcal{T}_g$ robots move to the same point and the $\mathcal{T}_{Apf}$ robots form the fixed pattern that is given as input.

\vspace{0.01\linewidth}
\begin{table}
\small
\begin{center}
\caption{If any of the Boolean variable $P_i$ where $0\le i \le 14$ on the left column is true then the corresponding condition on the right column is satisfied and vice versa.} 
 \label{TABLE-1}
\begin{tabular}{ | p{1.5em} | p{7.5cm}| } 
   \hline
   $P_0$ & $\mathcal{C}=\mathcal{C}_{final}$ \\
   \hline
    $P_1$ & $\mathcal{C}'=\mathcal{C}'_{final}$\\
    \hline
    $P_2$ &  All $\mathcal{T}_g$ robots are at a same point\\
    \hline
    $P_3$ &  $\mathcal{C}_{Apf}=\mathcal{C}_{target}$\\
    \hline
    $P_4$ & The current configuration is asymmetric \\
    \hline
    $P_5$ & $m = \max{\{N,n\}}+2$ \\
    \hline
     $P_6$ & $m = 2.\max{\{M,V\}}$ where $V$ is the length of the vertical side of the smallest enclosing rectangle of $\mathcal{C}'$ \\
    \hline
     $P_7$ & The $head$ in $\mathcal{C}$ is at the origin\\
    \hline
     $P_8$ & $n\geq \max{\{N+1,H+1,k\}}$ where $H$ is the length of the horizontal side of the smallest enclosing rectangle of $\mathcal{C}'$\\
    \hline
     $P_9$ & Line formation on $x$-axis without $tail$\\
    \hline
     $P_{10}$ & $\mathcal{C}'$ has a non-trivial reflectional symmetry with respect to a vertical line\\
    \hline
    $P_{11}$ & $\mathcal{C}'''=\mathcal{C}'''_{final}$\\
    \hline
    $P_{12}$ &  $\mathcal{C}''=\mathcal{C}''_{final}$ \\
    \hline
    $P_{13}$ &  $m\geq \max{\{N,n\}}+2$, $m \geq 2.\max{\{M,V\}}$  and $m$ is odd\\
    \hline
    $P_{14}$ & There exist a multiplicity point \\
    \hline
 
\end{tabular}
\vspace{0.01\linewidth}
\end{center}
\end{table}
\vspace{0.01\linewidth}

\subsection{Description of the Stages}
The main difficulty of this problem is a robot does not know which team another robot belongs to. The robot only knows about its own team. Here we will show that there will not arise any symmetry and no collision will occur. The global coordinate system will also not change. The algorithm is divided into eight stages. In this situation, the global coordinate will fix as we mention in sec~\ref{global}

\paragraph{Stage-1} \textit{Input:} $\{P_4 \land \neg P_{14} \land \neg (P_5 \land P_6)\}$ is true.\\
The $tail$ robot will move upwards and all other robots will remain static. \\
\textit{Aim:} The aim is to make ${\{P_5\land P_6}\}$ = true. After finite number of moves by $tail$ robot stage-1 completes with $\{P_4 \land \neg P_{14} \land P_5 \land P_6\}$ is true. 
\begin{theorem}
    If we have an asymmetric configuration $\mathcal{C}$ in stage 1 at some time $t$, then 
    \begin{enumerate}
        \item after one move by the tail towards upward, the new configuration is still asymmetric and the coordinate system remains unchanged.
        \item after a finite number of moves by the tail towards upward, stage~1 terminates with $(P_5 \land P_6)$= true.
    \end{enumerate}
\end{theorem}
\begin{proof}
    Let $\mathcal{C}$ be the asymmetric configuration where $ABCD$ be the smallest enclosing rectangle of the initial configuration and $|AB|=n$, $|AD|=m$, $m \ge n$. By one move of the tail robot towards upward, the new smallest enclosing rectangle is $ABC'D'$. Let $r$ be the tail robot and after the movement of $r$ it is now on the edge $C'D'$. As $\lambda_{AB}$ is the largest lexicographically string. So $\lambda_{AB} > \lambda_{BA}$. Now the new strings associated to the corner $A$ of the smaller side of $ABC'D'$ be $\lambda'_{AB}$ and $\lambda'_{BA}$.

    Let $r$ be the only robot on $CD$. Then it is easy to see that $\lambda'_{AB} > \lambda'_{BA}$. But when there are more than one robot on $CD$, then we have to show that by movement of tail robot $\lambda'_{AB} > \lambda'_{BA}$ is true.
    Let $r$ corresponds to the $i^{th}$ term and $j^{th}$ term of $\lambda_{AB}$ and $\lambda_{BA}$. If $i=j$ then $r$ is the middle robot of $CD$. Then when tail moves upward one step, the increase number of 0's of $\lambda'_{AB}$ and $\lambda'_{BA}$ are same. So $\lambda'_{AB} > \lambda'{BA}$. Now if $i<j$ then if we calculate the binary strings of $\lambda_{AB}$ and $\lambda_{BA}$ then $tail$ will be appear earlier in $BA$, than $AB$. Now if we calculate the binary strings of first $j$ term of AB and BA then $\lambda_{BA}|_j$ $>$ $\lambda'_{BA}|_j$. Also $\lambda_{AB}|_j$ $>$ $\lambda'_{AB}|_j$. But $\lambda_{AB}$ $>$ $\lambda_{BA}$. So we have $\lambda_{AB}|_j$ $>$ $\lambda_{BA}|_j$. Finally we can say $\lambda'_{AB}|_j$ $>$ $\lambda'_{BA}|_j$, so $\lambda'_{AB}$ $>$ $\lambda'_{BA}$. When $i<j$ in that case when $tail$ robot move upward  then in the new $\mathcal{SER}$ we can calculate that in this case also $\lambda'_{AB}$ and $\lambda'_{BA}$.

So in all the cases $\lambda'_{AB}$ $>$ $\lambda'_{BA}$. Now we show that the binary string of $AB$ is larger than $C'D'$. As we know that the binary string of $AB$ is larger than $CD$, but when the tail robot moves one step upward in that case as there is no robot other than the tail in $C'D'$ so if we calculate binary string it will be $\lambda'_{AB}$ $>$ $\lambda'_{C'D'}$.\\

In this case, $ABC'D'$ is a non-square grid, so four binary strings to consider here. By calculating we can say that $AB$ is the largest binary string in this new smallest enclosing rectangle. So the new configuration is still asymmetric. So the coordinate system is unchanged. As the tail moves upward so the $x$-coordinate remains unchanged.
    
\end{proof}
\paragraph{Stage-2} \textit{Input:} $\{P_4 \land \neg P_{14} \land P_5 \land P_6 \land \neg P_7\}$ is true.\\
The $head$ robot will move to origin. So $head$ robot will move towards left if it is not initially at the origin. In this move $head$ robot will remain $head$. \\
\textit{Aim:} After finite number of moves by the $head$ robot stage-2 completes when $\{P_4 \land \neg P_{14} \land P_5 \land P_6 \land P_7\}$ is true.
\begin{theorem}
    If we have an asymmetric configuration $\mathcal{C}$ in stage~2 at some time $t$, then
    \begin{enumerate}
        \item after one move by the head leftwards, the new configuration is still asymmetric and the coordinate system is unchanged.
        \item after finite number of moves by the head, $P_7$ becomes true, and Stage~2 terminates with $\{P_4 \land \neg P_{14} \land P_5 \land P_6 \land P_7\}$ true.
    \end{enumerate}
\end{theorem}
\begin{proof}
    Let $\mathcal{C}$ be the initial configuration and $ABCD$ be the smallest enclosing rectangle. As $\lambda_{AB}$ is the largest string so the position of head robot is the first 1 in the largest string. Let $i$ be the first position of 1 in $\lambda_{AB}$ string. So there is no 1 before $i^{th}$ position. Now when head robot moves left then the 1 will occur earlier than the $i^{th}$ position. So $\lambda_{AB}$ remains the lexicographically largest string upto head reaches origin. So 1) and 2) holds.
\end{proof}
 \begin{figure}[t!]
\centering
\includegraphics[width=.6\linewidth]{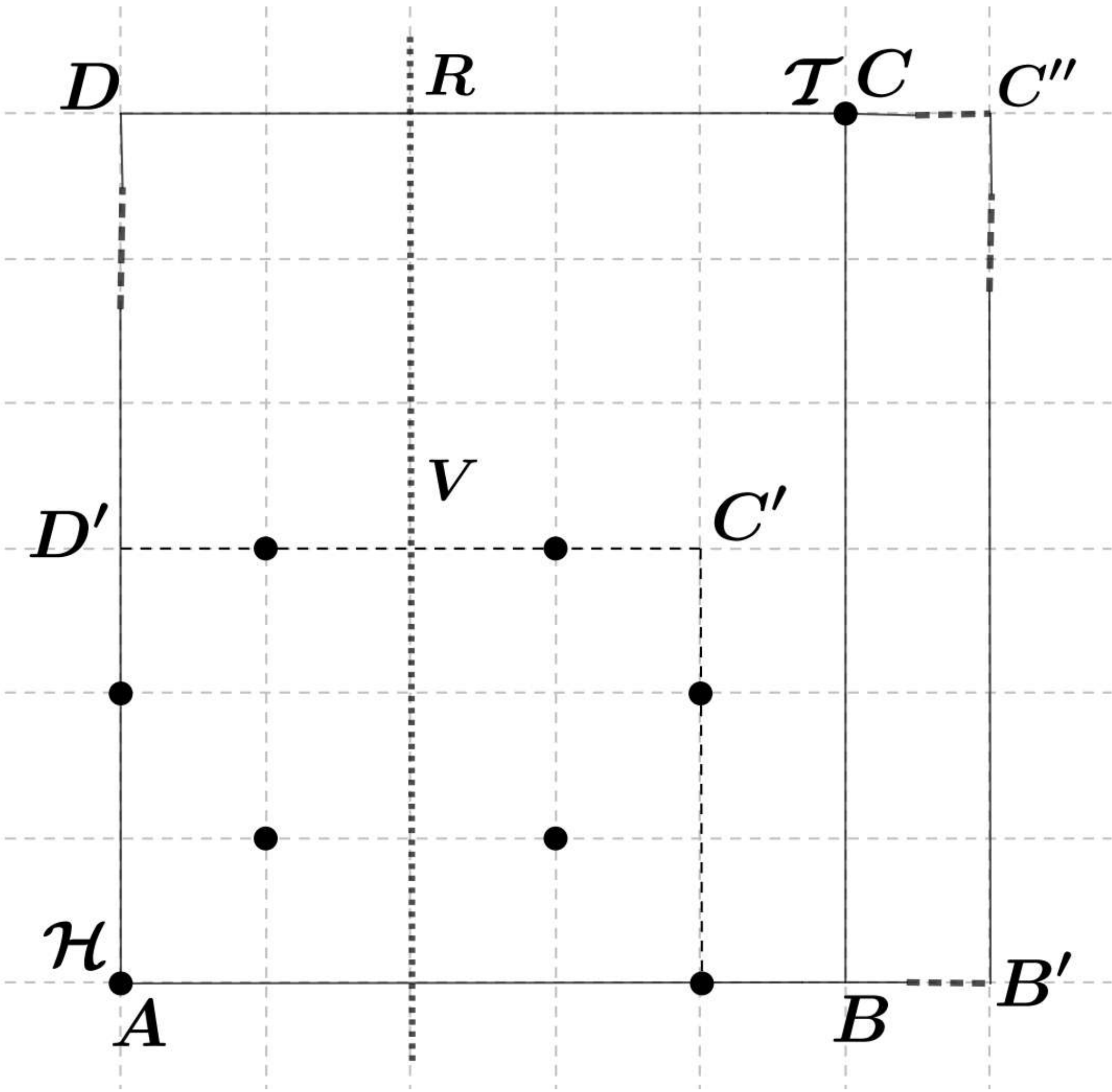}
     \caption{Case-1: $\mathcal{C'}$ has a vertical symmetry and $tail$ will move rightwards }
     \label{Fig:3a}
     \end{figure}
\begin{figure}[t!]
\centering
\includegraphics[width=.6\linewidth]{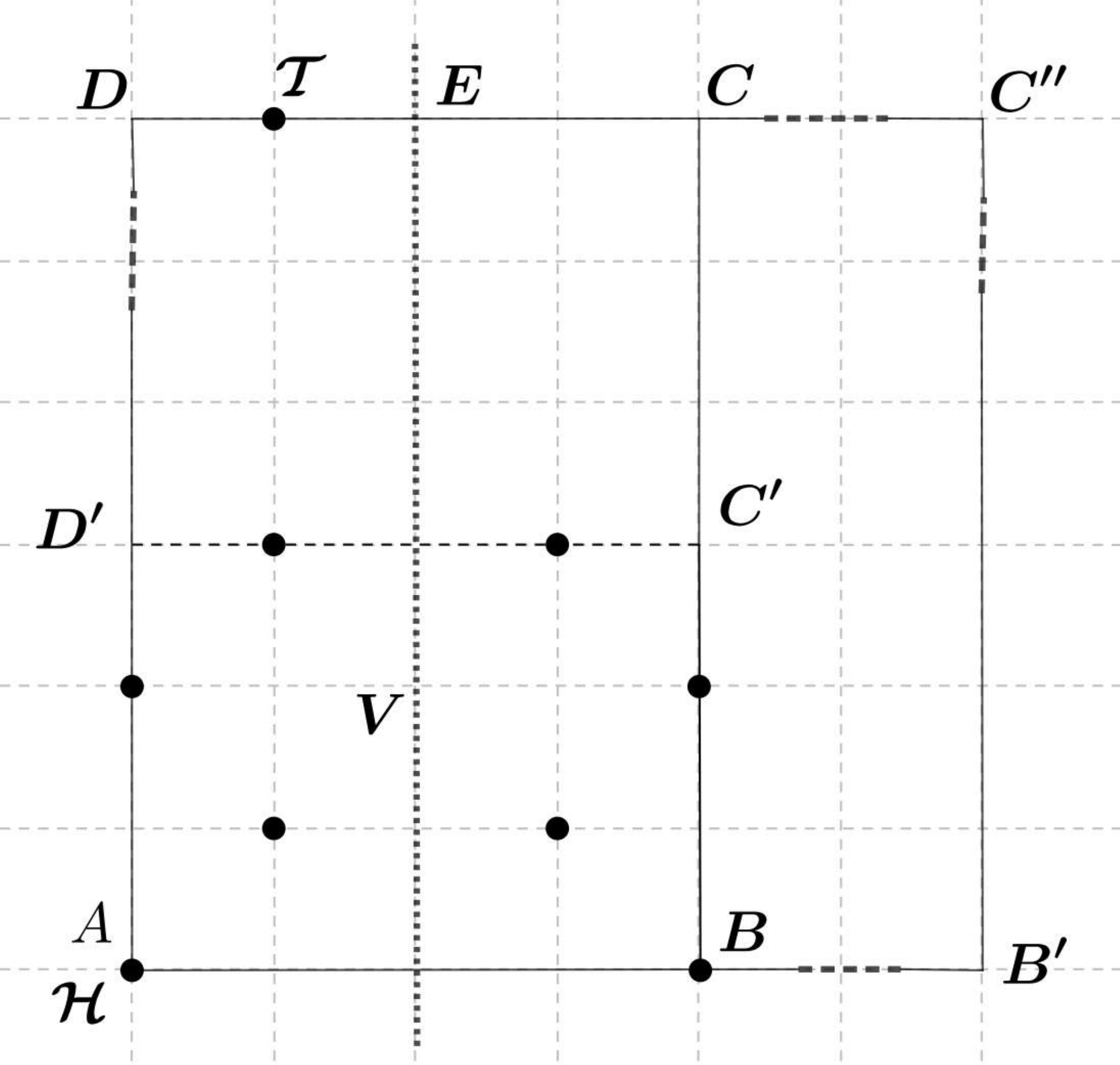}
    \caption{Case-2: $\mathcal{C'}$ has a vertical symmetry and $tail$ will move leftwards }
    \label{Fig:3.a}
\end{figure}

\paragraph{Stage-3} \textit{Input:} $\{P_4 \land \neg P_{14} \land P_5 \land P_6 \land P_7 \land \neg P_8\}$ is true.\\
The $tail$ robot will move rightwards.\\ 
\textit{Aim:} The main aim of this stage is to make $P_8$ true. After movement of robots $\{P_4 \land \neg P_{14} \land P_5 \land P_6 \land P_7 \land P_8\}$ is true.\\

\begin{theorem}
    If we have an asymmetric configuration $\mathcal{C}$ in stage~3 at some time $t$ with $P_{10}$ is false, then
    \begin{enumerate}
        \item after one move by the tail rightwards, the new configuration is still asymmetric and coordinate system is unchanged.
        \item after one move by the tail robot rightwards, $P_8$ becomes true.
        \item after finite number of moves by the tail, $\{P_4 \land \neg P_{14} \land P_5 \land P_6 \land P_7 \land P_8\}$ is true.
    \end{enumerate}
\end{theorem}
\begin{proof}
    Let $ABCD$ is the smallest enclosing circle at time $t$, where $\lambda_{AB}$ is the largest string. After one move by the tail rightward there may arise two cases.
    
\paragraph*{Case-1:} Let the smallest enclosing circle remain unchanged after one move by the tail rightward. As in this phase, the head is in origin and the tail has moved until $P_5$ and $P_6$ is true, in the binary string of $CD$ or $DC$ is smaller than $AB$. Also in this phase, $P_{10}$ is not true. So we must have $AB$ larger string than $BA$, so finally we get $\lambda'_{AB} > \lambda'_{BA}$.

\paragraph*{Case-2:} Suppose now tail robot is at $C$, then by one move of tail the new $\mathcal{SER}$ is $APQD$, where $|AP|$= $(n+1)$. Now it is easy to check that as $m\ge n+2$ so $m>n+1$. So we get that $AD>AP$. This implies that the new configuration is still not square, so we have to consider here only four binary strings, and as earlier $\lambda_{AP}$ will be a larger string. So we can conclude that the configuration is still asymmetric and the coordinate system is not changed by one move of the tail. It is easy to check that $P_6$ and $P_7$ are true here but not $P_{12}$. After the movement of the tail, $P_5$ may become false, so we are in stage 1 then. Then the tail moves upwards and one upwards move still has $P_5 \land P_6 \land P_7 \land  P_8 \neg P_{12}$ is true.\\

 Note that $P_{10}$ is either true or false in stage~3 by a finite number of moves of the tail the configuration remains asymmetric.
 
\end{proof}
In this stage the robots will check if $P_{10}$ is true or false. If $P_{10}$ is true then there may arise two cases:\\
\paragraph*{Case-1} If $tail$ is in the rightwards of the vertical symmetric line they it will move and make $P_8$ true.
\paragraph*{Case-2} If the $tail$ robot is in the left direction of the vertical symmetric line then it will not move rightwards. $Tail$ will then move in leftwards upto one more step than $D$ (fig~\ref{Fig:3.a}). In this case the co-ordinate system will be changed.
\begin{theorem}
    If we have an asymmetric configuration $\mathcal{C}$ in stage~3, then after finite number of moves by the tail robot, the configuration remains asymmetric.
\end{theorem}
\paragraph{Stage-4} \textit{Input:}$\{P_4 \land \neg P_{14} \land P_5 \land P_6 \land P_7 \land P_8 \land \neg P_9\}$ is true.\\
The inner robots will move to the x-axis and form a line. Other than the $tail$ robot, all other inner robots form a line on the $x$-axis. In stage four, the $head$ is in origin, and all the robots on the $x$-axis first make the line compact, i.e. there is no empty grid point between two robots. After the $x$-axis becomes compact, when a robot $r_i$ is on $H_i$ and there are no robots in between $H_i$ and the $x$-axis and the right part of $r_i$ is empty in its horizontal line, then the robot moves to the $x$-axis. This procedure is done one by one by robots. In between this movement, no collision will occur. So all the inner robots other than the $tail$ form a line on the $x$-axis.\\
\textit{Aim:} When all the inner robots move to x-axis then $\{P_4 \land \neg P_{14} \land P_5 \land P_6 \land P_7 \land P_8 \land P_9\}$ is true.

\begin{theorem}
    If we have an asymmetric configuration $\mathcal{C}$ at some time $t$, with $\{P_4 \land \neg P_{14} \land P_5 \land P_6 \land P_7 \land P_8 \land \neg P_9\}$ is true, then
    \begin{enumerate}
        \item after any move by the inner robots the configuration is asymmetric and coordinate system is unchanged.
        \item after finite number of moves by the inner robots $P_9$ becomes true and stage~4 terminates with with $\{P_4 \land \neg P_{14} \land P_5 \land P_6 \land P_7 \land P_8 \land P_9\}$ is true.
        \end{enumerate}
\end{theorem}
\begin{proof}
    In stage~4, head robot is in origin and tail robot is at a position satisfying $P_5 \land P_6 \land P_8$ is true. In this scenario, when robots on x axis move left then the string associated to $AB$ becomes larger than the previous one. Also in this stage one by one horizontal line's robots will move to the x axis. The movement of robots on a horizontal line will be from right to left. So in this movement by the inner robots, coordinate system will not change and as the movement is sequential so no collision will occur by the movement of inner robots. By finite moves by the robots all inner robots are on x axis without the tail robot. Then stage~4 terminates.
\end{proof}

\paragraph{Stage-5} \textit{Input:} $P_4 \land \neg P_{14} \land \neg P_{13} \land P_7 \land P_9$ is true.\\
One robot which is not on a compact line will move to its closest end point of that line following the shortest path if it $ \in \mathcal{T}_g$. As without one robot all other robots are on line so that one robot will move downwards. Then $P_{14}$ is true. Or when $(P_4 \land \neg P_{14} \land P_7 \land P_9)\land (P_5\land P_6\land P_8)$ is true then that one robot will move upward from its position if it $\in \mathcal{T}_{Apf}$. By this move $P_{13}$ will be true.\\

 \begin{figure}[t!]
\centering
\includegraphics[width=.75\linewidth]{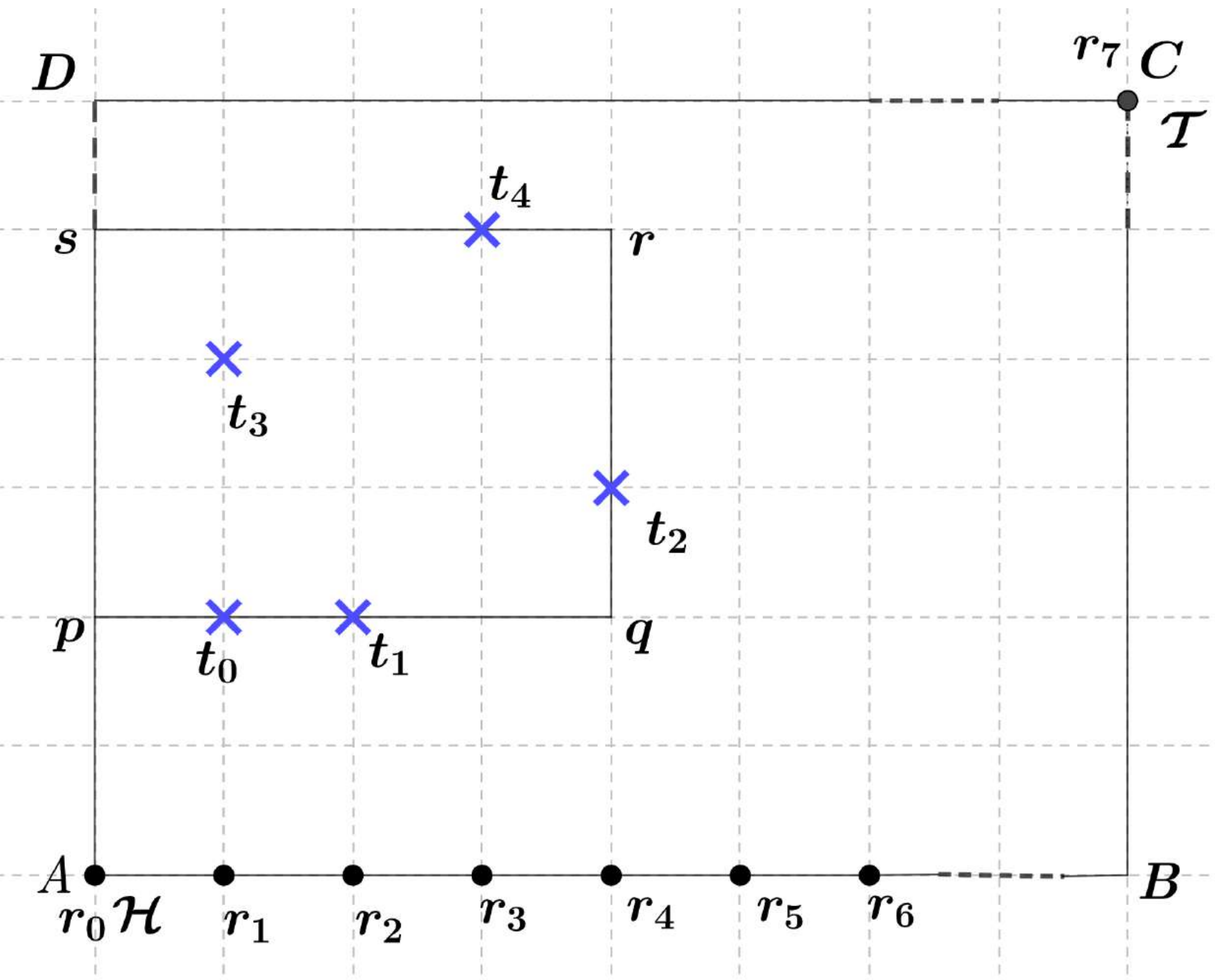}
     \caption{Any $r_i \in \mathcal{T}_g$ moves to $A$ and $r_i \in \mathcal{T}_{Apf}$ moves to fixed target positions one by one.}
     \label{Fig:3}
     \end{figure}
\begin{figure}[t!]
\centering
\includegraphics[width=.75\linewidth]{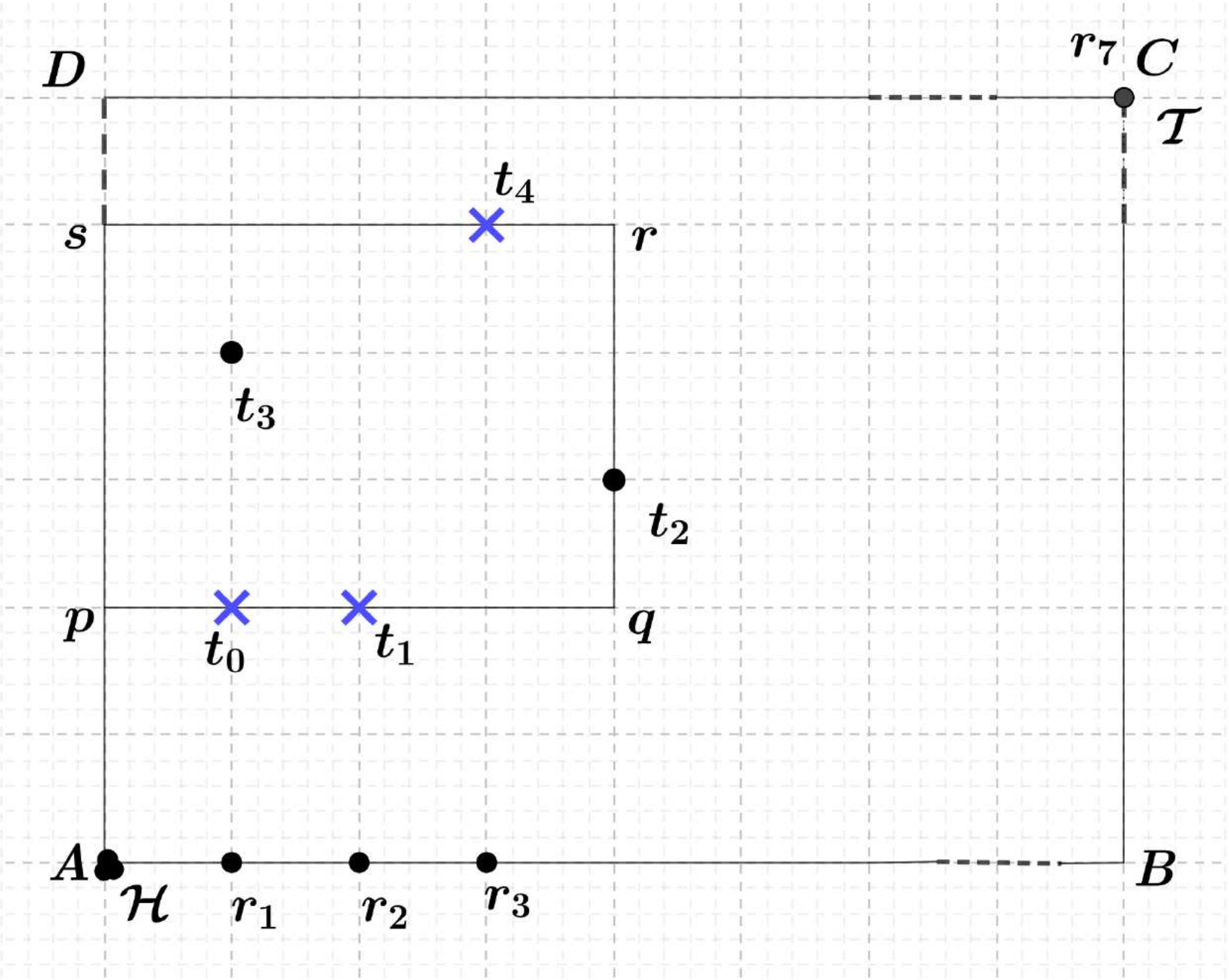}
    \caption{A multiplicity point at $A$ and some of the $\mathcal{T}_{Apf}$ robots form the pattern }
    \label{fig4}
\end{figure}

 \textit{Aim:} $P_4 \land \neg P_{14} \land P_7 \land P_9 \land (P_{13}\lor P_{14})$ is true..\\
After this stage if $P_{14}$ is true then all robots will fix the multiplicity point as origin, the compact line is as x axis and the other perpendicular axis as y axis.
\begin{theorem}
    If we have an asymmetric configuration $\mathcal{C}$ in stage~5, then after movement of one robot either one multiplicity point will create or $P_{13}$ will be true.

\end{theorem}
   \begin{proof}
       After the stage~4, all the inner robots are on x axis. When the tail robot is in team $\mathcal{T}_{Apf}$ then it will move upward one step and makes $P_{13}$ true. Without the tail all other robots are on line so head robot will not change and the $\lambda_{AB}$ remains the lexicographically largest string. So the coordinate system will not change. Now when the tail robot is in team $\mathcal{T}_{g}$ then it will move downwards. In this case also, tail will move downwards to its nearest corner. As all other robot is in x axis, when tail robot move downwards, the coordinate system does not change untill it reaches the nearest corner robot. But when tail robot reaches its nearest corner robot, then one multiplicity point will create. Afterwards the multiplicity point will be treated as origin and the line where other robots are placed treated as x axis. Any line perpendicular to the line treated as y axis. So in this way by the move of the tail robot either $P_{14}$ that is one multiplicity point or $P_{13}$ will be true.
   \end{proof} 

\paragraph{Stage-6} \textit{Input:} $P_{13} \lor P_{14}$ is true. \\
When the $tail$ robot will see that $P_{13}$ is true then it will not move.
An inner robot when sees that it is the rightmost robot on x-axis and $P_{13}$ is true and there exists robots at the positions $t_{k-1},t_{k-2}.....t_{k-i+1}$ where $1\le i\le k$ (let $t_k$ be the position of $tail$) then that inner robot will move upwards to $t_{k-i}$ if it $\in \mathcal{T}_{Apf}$. If the robot $\in \mathcal{T}_g$ then it will move to the origin. Note that in this case robot can fix the global coordinate so it will move to the origin by choosing the first horizontal line in the positive direction of the y-axis. When a rightmost robot of team \textsc{Apf} on x-axis sees that $P_{13}$ is true and no robot on target positions then it will move to $t_{k-1}$. 
Note that when a left most robot without $head$ on x axis of \textsc{Apf} team sees that only $tail_{target}$ or $tail_{target}$ and $head_{target}$ are not occupied by robots then the robot will move to $tail_{target}$. 
Again if $P_{14}$ is true then the rightmost robot on x-axis sees that there exist robots on $t_{k},t_{k-1}.....t_{k-i+1}$ where $1\le i\le k$ then that inner robot moves upward to $t_{k-i}$ if it $\in \mathcal{T}_{Apf}$. If it $\in \mathcal{T}_g$ then move to multiplicity point. Here when a robot can see all robots on a line then it will fix that line as the x-axis and anyone perpendicular line as the y-axis and the multiplicity point as the origin. So the robot $\in \mathcal{T}_g$ moves to the origin by choosing the first horizontal line in the positive direction of y-axis. In this way, the inner robots move to their target positions one by one.\\
\textit{Aim:} After this stage $P_{14} \land P_{12}$ is true.\\
After this stage, the multiplicity will be the origin and as the target pattern will be formed with respect to $(0,2)$ so either there will be a robot at $tail_{target}$ or the $tail$ robot will be at a position maintaining $P_{13}$ true. So we can choose the larger side as the y-axis and the other one as the x-axis.

\begin{figure}[t!]
 \centering
     \includegraphics[width=\linewidth]{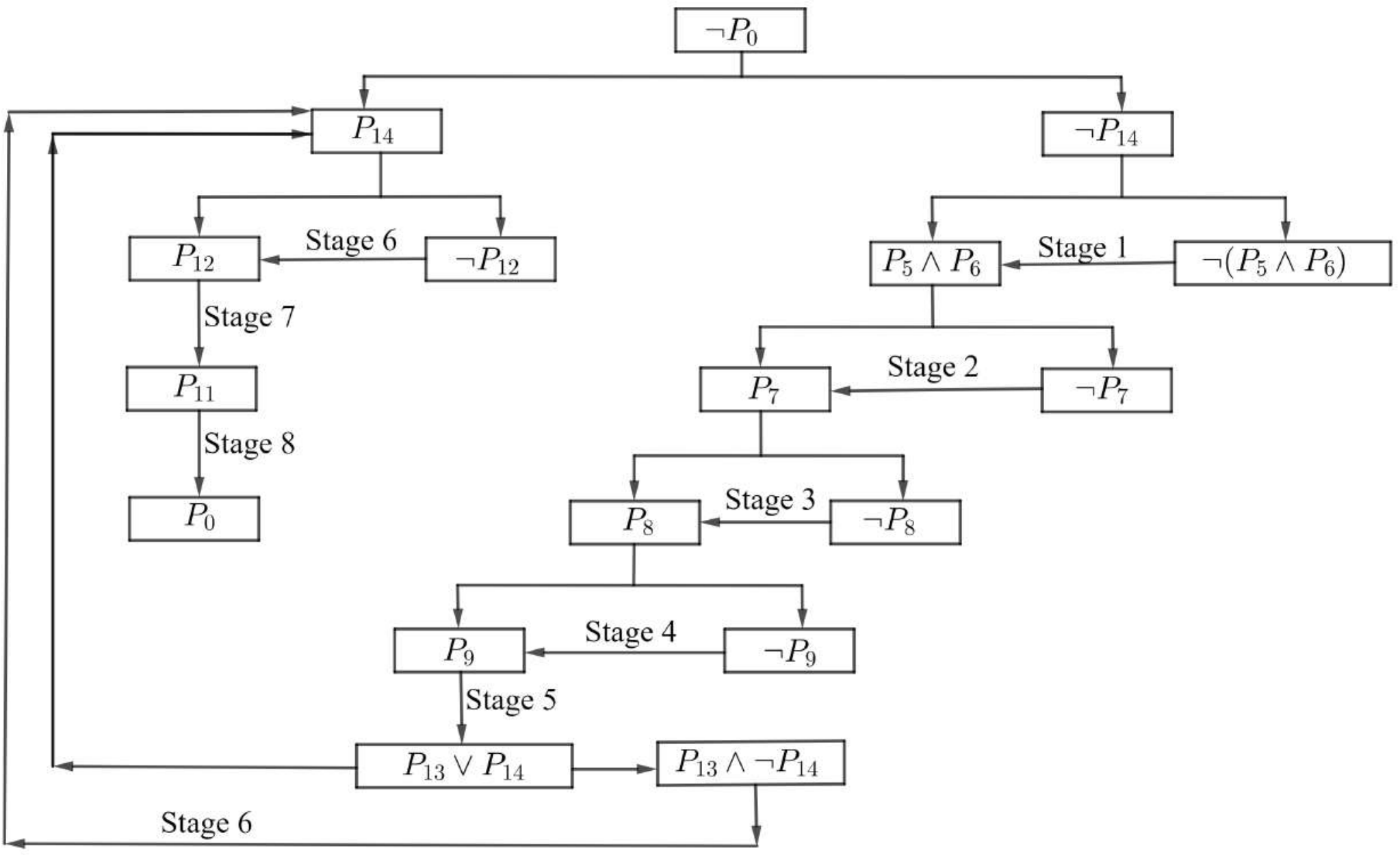}
     \caption{Flow chart of the algorithm}\label{flowchart1}
    \end{figure}

    \begin{theorem}
        If we have an asymmetric configuration $\mathcal{C}$ in stage~6 at time $t$, then after a finite number of moves by the inner robots stage~6 terminates and $P_{12}$ is true.
    \end{theorem}
    
 \begin{proof}
      In this stage, all the inner robots on the line x-axis will move to its destination point which will be either on the target position or the origin. Let, $P_{13}$ is true. In that case, the tail robot will be at a position maintaining the $P_{13}$ is true. So the inner robots on the x axis will move one by one. As $P_{13}$ is true, so the rightmost robot on x axis first moves either at $t_{k-1}$ if it is $\in \mathcal{T}_{Apf}$ or it will move to origin. The robot will move to the origin choosing the first horizontal line. As without tail, all robots are on x axis, and the target embedding is with respect to $(0,2)$ position. So there will be no robot on the first horizontal line above x axis. So the robots of x axis from right to left will move to the position, either at $t_{k-1}, t_{k-2}.....$ or move to origin. So within this sequential movements of robots, no collision will occur and configuration remain asymmetric. Now when $P_{14}$ is true. In that case, the multiplicity point will be origin and the robots on a line will be x axis. The perpendicular any direction of x axis will be y axis. So the rightmost robot on x axis will move $t_{k}$ if it is $\in \mathcal{T}_{Apf}$ or move to the origin. Here also, robots move one by one. So no collision and symmetric configuration will occur. If $P_{14}$ true and head and tail both are $\in \mathcal{T}_{g}$, then by this stage $P_{0}$ will be true. Also when tail is in $\mathcal{T}_{g}$ but head is in $\mathcal{T}_{Apf}$, then by this stage $P_{11}$ will be true and then it will be in stage~8. But if head and tail are in $\mathcal{T}_{Apf}$ and $P_{13}$, then by the movement of inner robots, $P_{12}$ will be true and stage~6 terminates.
 \end{proof}
    
\paragraph{Stage-7} \textit{Input:}In this stage $P_{14} \land P_{12} \land \neg P_{11}$ is true.\\
If $tail \in \mathcal{T}_{Apf}$ then move to the position of the $tail_{target}$. But if $tail \in \mathcal{T}_g$ then the $tail$ will move to the origin. If $tail \in \mathcal{T}_g$ then the $tail$ will move downwards up to the x-axis and then move towards the multiplicity point (By $P_{13}$ condition we can say that there will be no robot in this path of $tail$'s movement). \\
\textit{Aim:} $P_{11}$ is true.

\begin{theorem}
If we have an asymmetric configuration $\mathcal{C}$ in stage~7 at time $t$ then after the finite number of moves stage~7 terminates with $P_{11}$ is true.
\end{theorem}
\begin{proof}
    When in stage~7, without head and tail robots, all the other robots are in target positions. So if $tail \in \mathcal{T}_{Apf}$ then there is no robot on the position of $tail_{target}$. The tail robot will first moves left up to the x-coordinate of $tail_{target}$ and then move down towards $tail_{target}$. In this move, as $P_{14}$ already true, so nu symmetry will occur and only one robot will move in this stage. So collision will not occur also. When the tail robot reaches the $tail_{target}$, then the $P_{11}$ will be true.  But when the tail robot is in team gather and $P_{12}$ is true, then all inner robots form the pattern with respect to $(0,2)$. So there is no robot on the x axis. Tail robot first move downwards and then move left through the x axis towards the origin. In the path of tail robots move, there is no other robot will be present as $P_{14}$ and $P_{12}$ are true. So after the movement of tail robot, $P_{11}$ will be true and stage~7 terminates.
\end{proof}
\paragraph{Stage-8} \textit{Input:} $P_{11} \land P_{14}$ is true.\\
If $head$ robot is in gather team then it will not move. So then $P_0$ is true. But if it belongs to $\mathcal{T}_{Apf}$ then it moves to upward and then to the position of $head_{target}$. \\
\textit{Aim:} After the movement of $head$ then $P_0$ is true.\\

 \begin{theorem}
     If we have an asymmetric configuration $\mathcal{C}$ in stage~8 at time $t$ then by the move of head robot, stage~8 terminates and $P_0$ is true.
 \end{theorem}
\begin{proof}
    When in stage~8, then head robot if is in team gather then $P_0$ is already true. As head will not move in this case. But if head is in $\mathcal{T}_{Apf}$ then in this stage all the other robots are in their destination points without the head robot. As already $P_{14}$ is true, and there are $|\mathcal{T}_g| > 2$), so here by the move of head robot $P_{14}$ remains true. As the pattern will be form with respect to the $(0,2)$ point, so head will upward and then move to its destination point. If After the movement of head robot, $P_0$ becomes true. 
\end{proof}

So after completing these stages $P_2$ and $P_3$ conditions will be true. No collision will occur during the movement of robots throughout the algorithm. So the gathering and arbitrary pattern formation will be done by the robots simultaneously on an infinite grid by oblivious robots.\\
The proposed algorithm is depicted in the flowchart in fig~\ref{flowchart1}. Starting from any configuration where $\neg P_0$ from the diagram fig~\ref{flowchart1} each directed path starting from the node where $\neg P_0$ ends at the node where $P_0$ is true.
Hence we can conclude the theorem.

\begin{theorem}
Gathering and Arbitrary pattern formation are solvable in \textsc{Async} by $\mathcal{T}_{Apf}$ and $\mathcal{T}_g$ robots from any asymmetric initial configuration.
\end{theorem}

\section{Conclusion}
In this work, we claim that by our algorithm two different teams of robots can simultaneously gather and form an arbitrary pattern on an infinite grid. A robot only knows that in which team it belongs to between gathering and \textsc{Apf}. To our knowledge in a grid network, this is the first work where two different tasks are performed simultaneously by two different teams of robots. Here we assume the grid is infinite and the visibility is full. So for further work, we can extend this work by assuming limited visibility and also when the grid is finite. 

\section{Acknowledgement} 
The first author is supported by West Bengal State Government Fellowship Scheme 
[P-1/RS/23/2020]. The second author is supported by University Grants Commission (UGC), the Government of India [NTA Reference no.: 201610019984]
The third author is supported by University Grants Commission (UGC), the Government of India [1132/(CSIR-UGC NETJUNE2017]. The last author is supported by SERB India [Project file no.: MTR/2021/000835].

%
%
%
 \bibliographystyle{splncs04}
 \bibliography{name}

\begin{thebibliography}{10}
\providecommand{\url}[1]{\texttt{#1}}
\providecommand{\urlprefix}{URL }
\providecommand{\doi}[1]{https://doi.org/#1}

\bibitem{bhagat2020}
Bhagat, S., Chakraborty, A., Das, B., Mukhopadhyaya, K.: Gathering over meeting nodes in infinite grid. In: Conference on Algorithms and Discrete Applied Mathematics. pp. 318--330. Springer (2020)

\bibitem{BhagatS20}
Bhagat, S., Flocchini, P., Mukhopadhyaya, K., Santoro, N.: Weak robots performing conflicting tasks without knowing who is in their team. In: Mukherjee, N., Pemmaraju, S.V. (eds.) {ICDCN} 2020: 21st International Conference on Distributed Computing and Networking, Kolkata, India, January 4-7, 2020. pp. 29:1--29:6. {ACM} (2020). \doi{10.1145/3369740.3369794}, \url{https://doi.org/10.1145/3369740.3369794}

\bibitem{BoseAKS20}
Bose, K., Adhikary, R., Kundu, M.K., Sau, B.: Arbitrary pattern formation on infinite grid by asynchronous oblivious robots. Theor. Comput. Sci.  \textbf{815},  213--227 (2020). \doi{10.1016/j.tcs.2020.02.016}, \url{https://doi.org/10.1016/j.tcs.2020.02.016}

\bibitem{cicerone20}
Cicerone, S., Fonso, A.D., Stefano, G.D., Navarra, A.: Arbitrary pattern formation on infinite regular tessellation graphs. CoRR  \textbf{abs/2010.14152} (2020), \url{https://arxiv.org/abs/2010.14152}

\bibitem{cord2016}
Cord-Landwehr, A., Fischer, M., Jung, D., Meyer auf~der Heide, F.: Asymptotically optimal gathering on a grid. In: Proceedings of the 28th ACM Symposium on Parallelism in Algorithms and Architectures. pp. 301--312 (2016)

\bibitem{das2019}
Das, S., Giachoudis, N., Luccio, F.L., Markou, E.: Gathering of robots in a grid with mobile faults. In: International Conference on Current Trends in Theory and Practice of Informatics. pp. 164--178. Springer (2019)

\bibitem{d2012gathering}
d’Angelo, G., Stefano, G.D., Klasing, R., Navarra, A.: Gathering of robots on anonymous grids without multiplicity detection. In: International Colloquium on Structural Information and Communication Complexity. pp. 327--338. Springer (2012)

\bibitem{fischer2017}
Fischer, M., Jung, D., Meyer auf~der Heide, F.: Gathering anonymous, oblivious robots on a grid. In: International Symposium on Algorithms and Experiments for Sensor Systems, Wireless Networks and Distributed Robotics. pp. 168--181. Springer (2017)

\bibitem{GHOSH2022}
Ghosh, S., Goswami, P., Sharma, A., Sau, B.: Move optimal and time optimal arbitrary pattern formations by asynchronous robots on infinite grid. Int. J. Parallel Emergent Distributed Syst.  \textbf{38}(1),  35--57 (2023). \doi{10.1080/17445760.2022.2124411}, \url{https://doi.org/10.1080/17445760.2022.2124411}

\bibitem{pritam22}
Goswami, P., Ghosh, S., Sharma, A., Sau, B.: Gathering on an infinite triangular grid with limited visibility. CoRR  \textbf{abs/2204.14042} (2022). \doi{10.48550/arXiv.2204.14042}, \url{https://doi.org/10.48550/arXiv.2204.14042}

\bibitem{manash22}
Kundu, M.K., Goswami, P., Ghosh, S., Sau, B.: Arbitrary pattern formation by asynchronous opaque robots on infinite grid. CoRR  \textbf{abs/2205.03053} (2022). \doi{10.48550/arXiv.2205.03053}, \url{https://doi.org/10.48550/arXiv.2205.03053}

\bibitem{KunduGGS22}
Kundu, M.K., Goswami, P., Ghosh, S., Sau, B.: Arbitrary pattern formation by opaque fat robots on infinite grid. Int. J. Parallel Emergent Distributed Syst.  \textbf{37}(5),  542--570 (2022). \doi{10.1080/17445760.2022.2088750}, \url{https://doi.org/10.1080/17445760.2022.2088750}

\bibitem{Suzuki96}
Suzuki, I., Yamashita, M.: Distributed anonymous mobile robots - formation and agreement problems. In: Problems, in the Proceedings of the 3rd International Colloquium on Structural Information and Communication Complexity (SIROCCO '96. pp. 1347--1363 (1996)

\end{thebibliography}
%

\end{document}